\documentclass[conference]{IEEEtran}
\IEEEoverridecommandlockouts
\usepackage{amsmath,amssymb,amsfonts}
\usepackage{graphicx}
\usepackage{textcomp}
\usepackage{xcolor}
\usepackage{amsthm}
\usepackage{bm}        
\usepackage[font=footnotesize]{caption}

\newtheorem{assumption}{Assumption}
\newtheorem{theorem}{Theorem}
\newtheorem{remark}{Remark}

\usepackage{booktabs}
\usepackage{multirow}
\usepackage{array}
\usepackage{adjustbox}

\usepackage[colorlinks=true, linkcolor=blue, citecolor=blue, urlcolor=blue, bookmarks=false, pdfpagemode=UseNone]{hyperref}

\usepackage{graphicx}

\usepackage{todonotes}
\usepackage{marginnote}
\usepackage{pdfcomment}

\usepackage{makecell} 
\usepackage{enumitem} 

\newcommand{\m}{\boldsymbol}

\usepackage[linesnumbered,ruled,boxed]{algorithm2e} 
\renewcommand{\KwData}{\textbf{Input: }}
\renewcommand{\KwResult}{\textbf{Output: }}
\usepackage[letterpaper, top=1in, left=0.75in, right=0.75in, bottom=0.75in]{geometry}

\def\BibTeX{{\rm B\kern-.05em{\sc i\kern-.025em b}\kern-.08em
		T\kern-.1667em\lower.7ex\hbox{E}\ke
		rn-.125emX}}
\begin{document}
	\title{Dynamic Constraint Reconstruction Based Control Barrier Functions for Safety-Critical Control of High-Dimensional Manipulators: Extended Version}
	
	\author{
		Bingsheng Zhang$^{1}$,
		Shen Wang$^{1*}$,
		Qiang Wang$^{1}$,
		Muguo Du$^{2}$,
		Donghai Shi$^{2}$,
		Xiaofeng Tao$^{1}$
		\thanks{This work is supported by the National Natural Science Foundations of China under Grant 62203062, 
			and partly by the Fundamental Research Funds for the Central Universities under Grants 2242022k60006 and 2024RC08, and BUPT-Yaowu Tech cooperation. $^{*}$Corresponding author: Shen Wang (shen.wang@bupt.edu.cn). 
			
			$^{1}$Bingsheng Zhang, Shen Wang, Qiang Wang, and Xiaofeng Tao are with the National Engineering Research Center of Mobile Network Technologies, Beijing University of Posts and Telecommunications,
			Beijing 100876, China (e-mails: {zhangbingsheng, shen.wang, wangq, taoxf}@bupt.edu.cn).}%
		\thanks{$^{2}$Muguo Du and Donghai Shi are with the Yaowu (Shenzhen) Technology Co., Ltd., Shenzhen 518052, China (e-mails: {dumuguo, shidonghai}@yaowutech.com).}%
		
	}
	\maketitle

	\begin{abstract}
		Control barrier functions (CBFs) provide formal safety guarantees for constrained nonlinear systems, but their effectiveness relies on accurate system dynamics. In high-dimensional manipulators subject to unknown disturbances and model uncertainties, fixed safety constraints constructed from nominal dynamics may become inconsistent with the actual system behavior, leading to safety degradation or excessive conservatism. This paper proposes a dynamic constraint reconstruction based control barrier function (DCR-CBF) framework for safety-critical control of disturbed robotic manipulators. An extended state observer is employed to estimate lumped disturbances online, and the estimated disturbance is incorporated into high-order control barrier functions to reconstruct safety constraints according to the estimated true dynamics. To address estimation inaccuracies, a safety margin is introduced, and a sufficient condition is derived to guarantee forward invariance under bounded estimation errors. Simulation studies on a 4-DOF excavation manipulator demonstrate that the proposed DCR-CBF method achieves zero safety violation under strong unknown disturbances while significantly improving trajectory-tracking performance compared with standard and robust CBF methods.
	\end{abstract}


	\section{Introduction}\label{sec:intro}
	
High-dimensional robotic manipulators, such as excavation robots and industrial robotic arms, are increasingly deployed in complex and unstructured environments
 In these applications, the systems are required to operate under strict safety constraints, including joint-position and velocity limits, while being subjected to strong nonlinearities, external disturbances, and unmodeled dynamics. Ensuring safety under such uncertain operating conditions remains a fundamental challenge in robotic manipulation.

Control barrier functions (CBFs) have emerged as an effective framework for safety-critical control of constrained nonlinear systems~\cite{Wang2025,Yang2026}. In particular, high-order control barrier functions (HOCBFs) enable systematic handling of high-relative-degree safety constraints and have been widely applied in robotic systems~\cite{Lin2025}. Combined with optimization-based controllers such as model predictive control (MPC), HOCBFs provide a practical mechanism for real-time safety enforcement~\cite{Liu2025}.

Despite these advantages, most existing CBF-based approaches construct safety constraints based on nominal system dynamics, implicitly assuming model accuracy~\cite{Yang2026}. However, in high-dimensional manipulators subject to unknown disturbances and model uncertainties, the actual system dynamics may deviate significantly from the nominal model. As a result, the constructed safety constraints may become inconsistent with the true system behavior, which can lead to safety degradation, constraint violation, or excessive conservatism.

To mitigate this issue, robust CBF methods~\cite{Wang2023} incorporate disturbance bounds into the constraint design to ensure safety under worst-case conditions. While effective in guaranteeing safety, such approaches often introduce excessive conservatism, which significantly degrades control performance. Alternatively, learning-based methods~\cite{Xu2025,He2026} and disturbance-observer (DOB)-based approaches~\cite{arcos2025robust,Dong2026} have been integrated into CBF frameworks to improve robustness. However, these methods primarily focus on disturbance rejection or compensation while still relying on nominal-model-based constraint construction. Consequently, they do not fundamentally resolve the inconsistency between safety constraints and the actual system dynamics.

To reduce the dependence on accurate system models, the extended state observer (ESO) estimates the lumped disturbance as an augmented state, thereby providing stronger robustness against model uncertainties than conventional DOBs~\cite{Chen2023,Zhang2024ESO,Morton2025}. Existing studies have demonstrated the effectiveness of ESO in safety-critical control, including safety-margin adaptation based on disturbance-estimation convergence~\cite{wang2023eso} and compensation for delayed system responses through online estimation of dynamic errors~\cite{Zhang2024ESO}. Nevertheless, disturbance estimation remains inherently imperfect, and residual estimation errors may still introduce additional safety risks near the constraint boundary~\cite{CHENG2026106690,Zhou2025}.

Motivated by these observations, this paper addresses the safety control problem from the perspective of constraint consistency. Instead of relying on worst-case disturbance bounds, we propose a dynamic constraint reconstruction based control barrier function (DCR-CBF) framework, which reconstructs safety constraints online using disturbance estimation. Specifically, an extended state observer (ESO) is employed to estimate unknown disturbances and model uncertainties in real time, and the estimated disturbance is embedded into the HOCBF formulation for dynamic constraint reconstruction. In this way, the safety constraints evolve consistently with the actual system dynamics, thereby reducing conservatism while preserving safety, as illustrated in Fig.~\ref{fig:frame}.

	\begin{figure}[t]
		\centering
		\includegraphics[width=1\columnwidth]{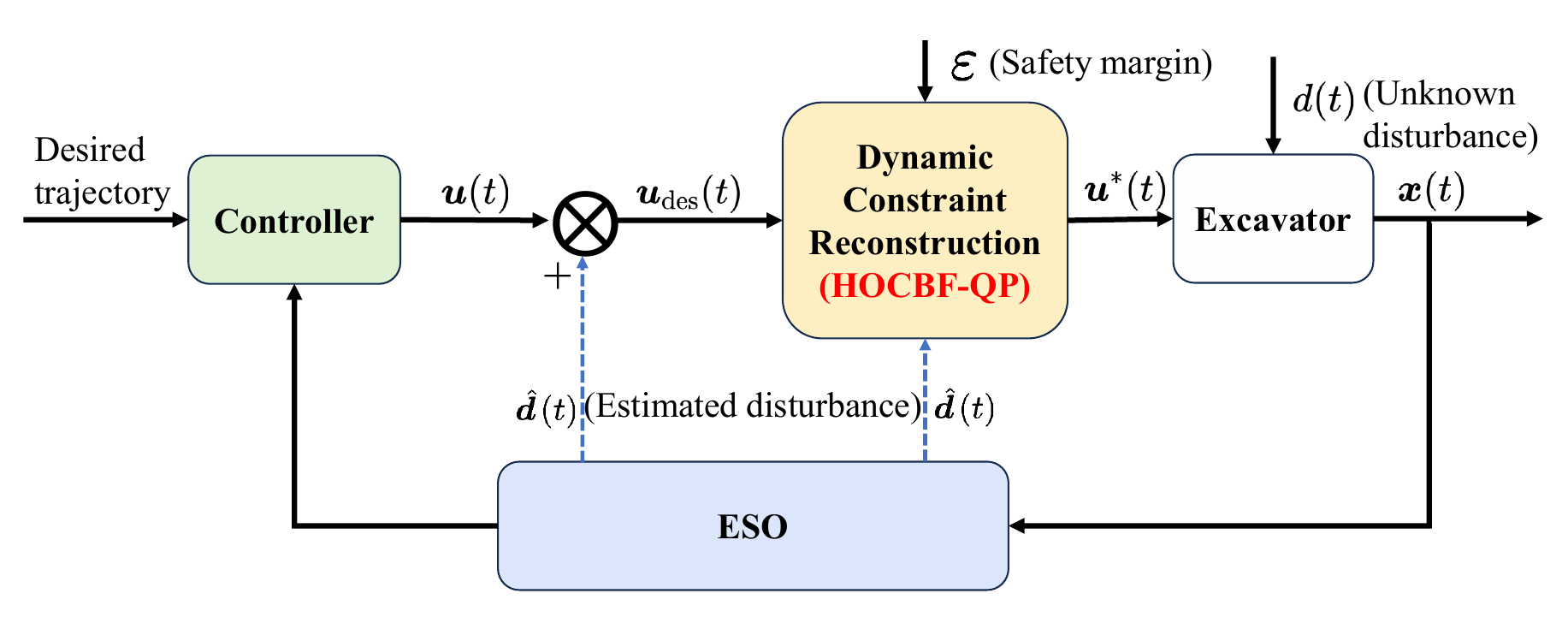}
		\caption{Architecture of the proposed DCR-CBF framework. The ESO estimates the lumped disturbance online and provides disturbance information for dynamic constraint reconstruction. The reconstructed HOCBF constraints are then used to generate safe control inputs $\m u^*$ for the excavation manipulator via solving a quadratic program (QP).}
		\label{fig:frame}
	\end{figure}
	
	Furthermore, to account for estimation inaccuracies, an error-aware safety margin is introduced. By explicitly considering the boundedness of the disturbance estimation error, a sufficient condition is established to guarantee forward invariance of the safe set under bounded uncertainties. 
	
	Based on the above observations, the main contributions of this paper are summarized as follows:
 
	\begin{itemize}
		
		\item 
		We identify constraint inconsistency under unknown disturbances as a fundamental limitation of conventional CBF-based safety control, and reformulate the safety problem from the perspective of dynamic constraint consistency.
		
		\item
		We develop a disturbance-aware HOCBF framework that integrates ESO-based disturbance estimation into online safety-constraint reconstruction, transforming fixed nominal constraints into adaptive real-time constraints and significantly reducing conservatism.
		
		\item
		We introduce a safety-margin design, establish a sufficient condition for forward invariance under bounded estimation errors, and experimentally validate the proposed framework on an excavation manipulator.
		
	\end{itemize}
	
	The remainder of this paper is organized as follows. Section~\ref{sec:ProblemFM} formulates the safety-control problem and analyzes the limitation of fixed-constraint CBF methods. Section~\ref{sec:dynamicCRF} presents the proposed dynamic constraint reconstruction framework. Section~\ref{sec:safetyCriticalCD} develops the safety-critical controller and provides the forward-invariance analysis. Section~\ref{sec:ExperimentalValidation} presents experimental validation. Section~\ref{sec:Conclusion} concludes this paper.
	
	\section{Problem Formulation and Motivation}\label{sec:ProblemFM}
	
	This section formulates the safety-control problem for high-dimensional robotic manipulators under unknown disturbances and highlights the fundamental limitation of conventional CBF-based methods. Specifically, we show that the main challenge lies not merely in control robustness, but in the inconsistency between fixed safety constraints and the actual disturbed system dynamics.

	\subsection{System Model and Safety Constraints}\label{sec:systemMSC}

	Consider an $n$-degree-of-freedom robotic manipulator described by
	\begin{equation}
		\m{M}(\m{q})\ddot{\m{q}}
		+
		\m{C}(\m{q},\dot{\m{q}})\dot{\m{q}}
		+
		\m{G}(\m{q})
		=
		\m{u}
		+
		\m{d}(t),
		\label{eq:real_dynamics}
	\end{equation}
	where $\m{q}\in\mathbb{R}^n$ and $\dot{\m{q}}\in\mathbb{R}^n$ denote the joint position and velocity vectors, respectively, $\m{u}\in\mathbb{R}^n$ is the control input, and $\m{d}(t)\in\mathbb{R}^n$ is a lumped disturbance.  The inertia matrix $\m{M}(\m{q})$, Coriolis term $\m{C}(\m{q},\dot{\m{q}})$, and gravity term $\m{G}(\m{q})$ are all state-dependent. To reduce the model's dependence on precise knowledge of these terms, an equivalent modeling strategy is adopted: all uncertainties arising from parameter variations and unmodeled dynamics, along with external disturbances, are absorbed into the disturbance term $\m{d}(t)$.
	
	Define the system state as $\m{x}=
	\begin{bmatrix}
		\m{q}^{\top},
		\dot{\m{q}}^{\top}
	\end{bmatrix}^{\top}.
	$ To ensure safe operation, each joint is required to satisfy the position constraints
	\begin{equation}
		\m{q}_{\min}
		\le
		\m{q}(t)
		\le
		\m{q}_{\max},
		\quad
		\forall t\ge0.
	\end{equation}
	
	Accordingly, the safe set is defined as
	\begin{equation}
		\mathcal{C}
		=
		\{
		\m{x}\in\mathbb{R}^{2n}
		\mid
		h(\m{x})\ge0
		\},
	\end{equation}
	where $h(\m{x})$ denotes the barrier function associated with joint constraints.
	
	The objective of this paper is to design a control law $\m{u}$ such that: (i) the safe set $\mathcal{C}$ remains forward invariant; (ii) the manipulator accurately tracks the desired trajectory; (iii) the safety constraints remain valid under unknown disturbances.

	\subsection{Limitation of Fixed-Constraint CBF}\label{sec:limitFCCBF}

	Conventional CBF methods construct safety constraints based on the nominal dynamics~\cite{Ames2017}
	\begin{equation}
		\m{M}(\m{q})\ddot{\m{q}}
		+
		\m{C}(\m{q},\dot{\m{q}})\dot{\m{q}}
		+
		\m{G}(\m{q})
		=
		\m{u},
		\label{eq:nominal_model}
	\end{equation}
	where the disturbance term $\m{d}(t)$ is neglected.
	
	For high-relative-degree constraints, the corresponding HOCBF condition~\cite{Lin2025} can generally be written as
	\begin{equation}
		\Psi
		(
		\m{x},
		\m{u};
		f_0
		)
		\ge0,
		\label{eq:nominal_constraint}
	\end{equation} where $f_0$ denotes the nominal dynamics.
	
	However, the actual system evolves according to Equation~\eqref{eq:real_dynamics}, which can be equivalently represented as $
	f_{\mathrm{real}} = f_0 + \Delta f(\m{d}).
	$ Consequently, the actual safety evolution becomes$
		\Psi
		(
		\m{x},
		\m{u};
		f_{\mathrm{real}}
		)
		\neq
		\Psi
		(
		\m{x},
		\m{u};
		f_0
		),
	$
	which implies that the fixed constraints constructed from nominal dynamics may become inconsistent with the actual system behavior.

	Such inconsistency introduces two fundamental issues: (i) constraint violation, where the nominally safe control input may fail to preserve safety under disturbances, and (ii)  excessive conservatism, where robust CBF methods guarantee safety by considering worst-case disturbance bounds but often sacrifice control performance. Therefore, the key challenge is not merely to design a more robust controller, but to make the safety constraints themselves consistent with the actual system dynamics. Motivated by this observation, this paper proposes an online \emph{dynamic constraint reconstruction} framework, in which safety constraints are continuously updated according to the estimated true system dynamics.

	\section{Dynamic Constraint Reconstruction Framework}\label{sec:dynamicCRF}
	
	This section presents the proposed dynamic constraint reconstruction framework. Instead of treating safety constraints as fixed objects derived from nominal dynamics, the proposed method reconstructs the constraint evolution online using disturbance estimation.

	\subsection{Core Idea of Constraint Reconstruction}\label{sec:coreIdeaCR}

	The central idea of this work is to reconstruct safety constraints according to the estimated true system dynamics. Conventional CBF methods enforce $ h(\m{x};f_0),$ where the barrier evolution is determined by the nominal dynamics. In contrast, we reconstruct the barrier dynamics as $
	h(\m{x};f_0)
	\rightarrow
	h(\m{x};\hat{f}), $ where $\hat{f}$ denotes the online estimated system dynamics. Since the disturbance is estimated in real time, the safety constraints become adaptive, i.e., $
	h(\m{x})
	\rightarrow
	h(\m{x},\hat{\m{d}}).
	$ In this way, the safety constraints evolve consistently with the actual system behavior, thereby reducing conservatism while preserving safety.

	\subsection{ESO-Based Disturbance Estimation}\label{sec:esoBaseDL}

	To reconstruct the safety constraints according to the disturbed system behavior, the lumped generalized disturbance in~\eqref{eq:real_dynamics} is estimated online. Let
	$
	\m{x}
	=
	[\m{q}^{\top},\dot{\m{q}}^{\top}]^{\top}
	\in\mathbb{R}^{2n}.
	$
	The manipulator dynamics can be written exactly in a state-dependent coefficient form as
	\begin{equation}
		\dot{\m{x}}
		=
		\m{A}(\m{x})\m{x}
		+
		\m{B}(\m{x})\m{u}
		+
		\m{r}(\m{x})
		+
		\m{E}(\m{x})\m{d}(t),
		\label{eq:state_space_disturbed}
	\end{equation}
	where
	\begin{equation}
		\begin{aligned}
			\m{A}(\m{x})
			&=
			\begin{bmatrix}
				\m{0}_{n} & \m{I}_{n}\\
				\m{0}_{n} & -\m{M}^{-1}(\m{q})\m{C}(\m{q},\dot{\m{q}})
			\end{bmatrix}
			,\\
			\m{B}(\m{x})
			&=
			\m{E}(\m{x})
			=
			\begin{bmatrix}
				\m{0}_{n}\\
				\m{M}^{-1}(\m{q})
			\end{bmatrix}
			,\\
			\m{r}(\m{x})
			&=
			\begin{bmatrix}
				\m{0}_{n\times 1}\\
				-\m{M}^{-1}(\m{q})\m{G}(\m{q})
			\end{bmatrix}
			.
		\end{aligned}
		\label{eq:state_dependent_matrices}
	\end{equation}

	Here, $\m{M}$, $\m{C}$, and $\m{G}$ denote the nominal model used by the controller; parameter variations and unmodeled dynamics are included in $\m{d}(t)$ together with the external disturbance. The dimensions of all matrices and blocks are explicitly indicated in \eqref{eq:state_dependent_matrices}. In particular, the disturbance distribution matrix is state dependent and is identical to the input distribution matrix because the lumped disturbance and the actuator torque enter through the same channel.
	
	The disturbance is introduced as the augmented state
	$
	\m{z}=\m{d}(t)
	$
	and
	$
	\m{x}_e=[\m{x}^{\top},\m{z}^{\top}]^{\top}\in\mathbb{R}^{3n}.
	$
	Defining $\m{v}(t)=\dot{\m{d}}(t)$, the augmented dynamics are
	\begin{equation}
		\dot{\m{x}}_e
		=
		\m{A}_e(\m{x})\m{x}_e
		+
		\m{B}_e(\m{x})\m{u}
		+
		\m{r}_e(\m{x})
		+
		\m{D}_e\m{v}(t),
		\label{eq:augmented_system}
	\end{equation}
	where
	\begin{equation}
		\begin{aligned}
			\m{A}_e(\m{x})
			&=
			\begin{bmatrix}
				\m{A}(\m{x}) & \m{E}(\m{x})\\
				\m{0}_{n\times 2n} & \m{0}_{n}
			\end{bmatrix},
			&
			\m{B}_e(\m{x})
			&=
			\begin{bmatrix}
				\m{B}(\m{x})\\
				\m{0}_{n}
			\end{bmatrix},\\
			\m{r}_e(\m{x})
			&=
			\begin{bmatrix}
				\m{r}(\m{x})\\
				\m{0}_{n\times 1}
			\end{bmatrix},
			&
			\m{D}_e
			&=
			\begin{bmatrix}
				\m{0}_{2n\times n}\\
				\m{I}_{n}
			\end{bmatrix}.
		\end{aligned}
		\label{eq:augmented_matrices}
	\end{equation}
	
	\begin{assumption}~\label{assumption1}
		The lumped disturbance $\m{d}(t)$ is continuously differentiable and its rate is bounded, i.e.,
		$
		\|\m{v}(t)\|=\|\dot{\m{d}}(t)\|\le\bar d,
		$
		where $\bar d>0$ is known. The manipulator state remains in a compact operating set $\mathcal X$ over which the matrices in \eqref{eq:state_dependent_matrices} are bounded.
	\end{assumption}
	
	For digital implementation, the state-dependent matrices are evaluated at the measured state at each sampling instant and held constant over one sampling interval. Applying a standard zero-order-hold discretization to \eqref{eq:augmented_system} gives
	\begin{equation}
		\m{x}_{e,k+1}
		=
		\m{A}_{d,k}\m{x}_{e,k}
		+
		\m{B}_{d,k}\m{u}_k
		+
		\m{r}_{d,k}
		+
		\m{w}_k,
		\label{eq:discrete_augmented}
	\end{equation}
	where $\m{A}_{d,k}$, $\m{B}_{d,k}$, and $\m{r}_{d,k}$ are the discrete counterparts of $\m{A}_e(\m{x}_k)$, $\m{B}_e(\m{x}_k)$, and $\m{r}_e(\m{x}_k)$, respectively. The bounded term $\m{w}_k$ collects the disturbance variation and the residual introduced by freezing the state-dependent matrices over one sampling interval.
	
	The joint-position vector is used as the ESO output, while both joint position and velocity remain available to the safety controller. Thus,
	\begin{equation}
		\m{y}_k
		=
		\m{q}_k
		=
		\m{C}_e\m{x}_{e,k},
		\qquad
		\m{C}_e
		=
		\begin{bmatrix}
			\m{I}_{n} & \m{0}_{n\times 2n}
		\end{bmatrix}
		\in\mathbb{R}^{n\times 3n}.
		\label{eq:measurement_model}
	\end{equation}
	Based on \eqref{eq:discrete_augmented}, the discrete-time ESO is constructed as
	\begin{equation}
		\hat{\m{x}}_{e,k+1}
		=
		\m{A}_{d,k}\hat{\m{x}}_{e,k}
		+
		\m{B}_{d,k}\m{u}_k
		+
		\m{r}_{d,k}
		+
		\m{L}_k
		\left(
		\m{y}_k-\m{C}_e\hat{\m{x}}_{e,k}
		\right),
		\label{eq:eso_discrete}
	\end{equation}
	where $\m{L}_k\in\mathbb{R}^{3n\times n}$ is the observer gain, selected according to the standard ESO gain-design procedure in~\cite{Chen2023}.
	The disturbance estimate is extracted from the augmented state as
	\begin{equation}
		\hat{\m{d}}_k
		=
		\begin{bmatrix}
			\m{0}_{n\times 2n} & \m{I}_{n}
		\end{bmatrix}
		\hat{\m{x}}_{e,k}.
		\label{eq:disturbance_extraction}
	\end{equation}
	Under Assumption~\ref{assumption1} and a stable observer-gain selection, the disturbance-estimation error remains uniformly ultimately bounded. The estimate $\hat{\m{d}}_k$ is subsequently used to reconstruct the HOCBF constraints.

	\subsection{Reconstruction of High-Order Safety Constraints}\label{sec:reconstructionHOSC}

	Based on the online disturbance estimation obtained from the ESO, the nominal safety constraints can be reconstructed according to the estimated true system dynamics. For robotic manipulators, the joint-position constraints are defined as
	$
	\m{q}_{\min}
	\le
	\m{q}(t)
	\le
	\m{q}_{\max}.
	$ Consider the upper-bound constraint. The corresponding barrier function is defined as
	\begin{equation}
		h(\m{x})
		=
		\m{q}_{\max}
		-
		\m{q}.
		\label{eq:joint_barrier}
	\end{equation}
	All vector inequalities in the following are interpreted componentwise.
	
	Since the control input does not explicitly appear in $\dot h(\m{x})$, the barrier function has relative degree two with respect to the control input. Under the nominal manipulator dynamics \eqref{eq:nominal_model}, the joint acceleration is
	\begin{equation}
		\ddot{\m{q}}
		=
		\m{M}^{-1}(\m{q})
		\Bigl(
		\m{u}
		-
		\m{C}(\m{q},\dot{\m{q}})\dot{\m{q}}
		-
		\m{G}(\m{q})
		\Bigr).
		\label{eq:nominal_acc}
	\end{equation}
	
	Let $\m{\lambda}_1,\m{\lambda}_2\in\mathbb R_{>0}^{n}$ denote the joint-wise gains of the two recursive HOCBF layers. The conventional HOCBF condition is~\cite{wang2023eso}
	\begin{equation}
		\ddot{h}
		+
		(\m{\lambda}_1+\m{\lambda}_2)\odot\dot{h}
		+
		(\m{\lambda}_1\odot\m{\lambda}_2)\odot h
		\ge0,
		\label{eq:nominal_hocbf}
	\end{equation}
	where the $\odot$ denotes the Hadamard product.
	
	Substituting \eqref{eq:joint_barrier} and \eqref{eq:nominal_acc} into \eqref{eq:nominal_hocbf} yields the nominal constraint
	$
	\m{M}^{-1}(\m{q})\m{u}
	\le
	\m{\eta}_0^{+}(\m{q},\dot{\m{q}}),
	$
	where $\m{\eta}_0^{+}(\m{q},\dot{\m{q}})$ denotes the nominal nonlinear term,  defined explicitly as follows 
	\begin{equation*}
		\begin{aligned}
			\m{\eta}_0^{+}(\m{q},\dot{\m{q}})
			={}&
			\m{M}^{-1}(\m{q})
			\bigl(\m{C}(\m{q},\dot{\m{q}})\dot{\m{q}}+\m{G}(\m{q})\bigr)
			\\
			&-(\m{\lambda}_1+\m{\lambda}_2)\odot\dot{\m{q}}
			+(\m{\lambda}_1\odot\m{\lambda}_2)\odot(\m{q}_{\max}-\m{q}).
		\end{aligned}
	\end{equation*}
	Note that $\m{M}^{-1}(\m{q})$ may have negative entries, so the inequality cannot be manipulated by multiplying both sides by $\m{M}$.
	
	However, under the actual disturbed dynamics \eqref{eq:real_dynamics}, the true joint acceleration becomes
	\begin{equation}
		\ddot{\m{q}}
		=
		\m{M}^{-1}(\m{q})
		\Bigl(
		\m{u}
		+
		\m{d}(t)
		-
		\m{C}(\m{q},\dot{\m{q}})\dot{\m{q}}
		-
		\m{G}(\m{q})
		\Bigr),
	\end{equation}
	which generally differs from the nominal model \eqref{eq:nominal_acc}. To make the barrier evolution consistent with the actual system behavior, the disturbance estimate $\hat{\m{d}}(t)$ obtained from the ESO is incorporated into the acceleration dynamics, yielding the reconstructed dynamics:
	\begin{equation}
		\ddot{\m{q}}
		=
		\m{M}^{-1}(\m{q})
		\Bigl(
		\m{u}
		+
		\hat{\m{d}}(t)
		-
		\m{C}(\m{q},\dot{\m{q}})\dot{\m{q}}
		-
		\m{G}(\m{q})
		\Bigr).
		\label{eq:reconstructed_acc}
	\end{equation}
	
	Substituting \eqref{eq:reconstructed_acc} into the HOCBF condition \eqref{eq:nominal_hocbf}, the reconstructed safety constraint is obtained as~\cite{Zhang2024ESO}
	\begin{equation}
		\m{M}^{-1}(\m{q})\m{u}
		\le
		\m{\eta}_0^{+}(\m{q},\dot{\m{q}})
		-
		\m{M}^{-1}(\m{q})\hat{\m{d}}(t).
		\label{eq:reconstructed_constraint}
	\end{equation}
	For the lower-bound barrier $h(\m{x})=\m{q}-\m{q}_{\min}$, the same derivation gives
	\begin{equation}
		\m{M}^{-1}(\m{q})\m{u}
		\ge
		\m{\eta}_0^{-}(\m{q},\dot{\m{q}})
		-
		\m{M}^{-1}(\m{q})\hat{\m{d}}(t),
		\label{eq:reconstructed_lower_constraint}
	\end{equation}
	where
	\begin{equation*}
		\begin{aligned}
		\m{\eta}_0^{-}(\m{q},\dot{\m{q}})
		={}&
		\m{M}^{-1}(\m{q})
		\bigl(\m{C}(\m{q},\dot{\m{q}})\dot{\m{q}}+\m{G}(\m{q})\bigr)
		\\
		&-(\m{\lambda}_1+\m{\lambda}_2)\odot\dot{\m{q}}
		-(\m{\lambda}_1\odot\m{\lambda}_2)\odot(\m{q}-\m{q}_{\min}).
		\end{aligned}
	\end{equation*}
	The upper and lower reconstructed constraints jointly define $\mathcal U_{\mathrm{safe}}(t)$ in the QP, whereas $\m{u}_{\min}$ and $\m{u}_{\max}$ remain the physical actuator limits.
	
	Compared with the conventional formulation $
	\m{M}^{-1}(\m{q})\m{u}
	\le
	\m{\eta}_0^{+}(\m{q},\dot{\m{q}})
	$, the disturbance estimate is directly embedded into the second-order Lie derivative of the barrier dynamics.
	
	Consequently, the HOCBF constraints are transformed from fixed nominal constraints into adaptive real-time constraints $
	h(\m{x};f_0)
	\rightarrow
	h(\m{x};\hat f).
	$ Therefore, the safety constraints evolve according to the estimated true system dynamics rather than fixed nominal dynamics, thereby significantly reducing the conservatism introduced by worst-case robust formulations.

	\section{Safety-Critical Controller Design}\label{sec:safetyCriticalCD}
	
	Based on the reconstructed HOCBF constraints derived in Section~\ref{sec:dynamicCRF}, this section develops the corresponding safety-critical controller and provides the safety guarantee under bounded disturbance-estimation errors.

	\subsection{QP-Based Safe Control with Reconstructed Constraints}\label{sec:QPbaseSCRC}

	Let $\m{u}_{\mathrm{des}}(t)\in\mathbb{R}^{n}$ denote the nominal tracking controller designed for task execution, shown in Fig.~\ref{fig:frame}, such as trajectory tracking or motion planning. To minimally modify the nominal control input while guaranteeing safety, the final control command is obtained by solving the following quadratic program (QP)~\cite{Ames2017} \cite{excavatorModle}:
	\begin{equation}
		\begin{aligned}
			\m{u}^{*}(t)
			=
			\arg\min_{\m{u}}
			\quad
			&
			\|
			\m{u}
			-
			\m{u}_{\mathrm{des}}(t)
			\|^2
			\\
			\textnormal{s.t.}
			\quad
			&
			\m{u}_{\min}
			\le
			\m{u}
			\le
			\m{u}_{\max},
			\\
			&
			\m{u}
			\in
			\mathcal U_{\mathrm{safe}}(t),
		\end{aligned}
		\label{eq:qp_controller}
	\end{equation}
	where
	$\mathcal U_{\mathrm{safe}}(t)$
	denotes the admissible control set induced by the reconstructed HOCBF constraints, $\m{u}_{\min}$ and $\m{u}_{\max}$ denote the actuator limits.
	
	Compared with conventional HOCBF-QP controllers~\cite{Lin2025}, the optimization structure in \eqref{eq:qp_controller} remains unchanged. The essential difference lies in the online reconstruction of the feasible safety set according to the estimated true system dynamics. Therefore, the feasible set of the QP evolves adaptively according to the estimated true system dynamics rather than remaining fixed under the nominal model.

	\subsection{Safety Margin Design}\label{sec:redundantSMD}

	Although the ESO provides online disturbance estimation, the estimation error cannot be completely eliminated. Define the disturbance-estimation error as
	$
	\m{e}_d(t)=\m{d}(t)-\hat{\m{d}}(t).
	$
	To compensate for the resulting uncertainty at both joint boundaries, introduce a joint-wise safety-margin vector $\m{\varepsilon}\in\mathbb R_{>0}^{n}$ and contract the admissible joint-position interval as
	\begin{equation}
		\m{q}_{\min}+\m{\varepsilon}
		\le
		\m{q}
		\le
		\m{q}_{\max}-\m{\varepsilon}.
		\label{eq:margin_barrier}
	\end{equation}
	Thus, the $i$th joint maintains a distance $\varepsilon_i$ from both its lower and upper limits.

	Applying the reconstructed HOCBF conditions \eqref{eq:reconstructed_constraint} and \eqref{eq:reconstructed_lower_constraint} to the contracted interval gives
 	\begin{subequations}
		\label{eq:robust_constraint}
		\begin{align}
			\m{M}^{-1}(\m{q})\m{u}
			&\le \m{\eta}_0^{+}(\m{q},\dot{\m{q}})
			- \m{M}^{-1}(\m{q})\hat{\m{d}}(t)
			\notag\\
			&\quad
			- (\m{\lambda}_1 \odot \m{\lambda}_2)
			\odot \m{\varepsilon},
			\label{eq:robust_constraint_a}\\[1ex]
			\m{M}^{-1}(\m{q})\m{u}
			&\ge \m{\eta}_0^{-}(\m{q},\dot{\m{q}})
			- \m{M}^{-1}(\m{q})\hat{\m{d}}(t)
			\notag\\
			&\quad
			+ (\m{\lambda}_1 \odot \m{\lambda}_2)
			\odot \m{\varepsilon}.
			\label{eq:robust_constraint_b}
		\end{align}
	\end{subequations}
	
	The first inequality tightens the upper-bound constraint downward, whereas the second tightens the lower-bound constraint upward. Consequently, both inequalities move toward the interior of the original joint interval.

	Therefore, the admissible control set in \eqref{eq:qp_controller} is updated as
	\begin{equation}
		\mathcal U_{\mathrm{safe}}(t)
		:=
		\left\{
		\m{u}\in\mathbb R^n
		\mid
		\eqref{eq:robust_constraint}
		\right\}.
	\end{equation}
	This symmetric boundary contraction compensates for possible deviations caused by disturbance-estimation errors near either joint limit.


	\subsection{Forward Invariance Analysis}\label{sec:fowardIA}

	To account for the initial ESO transient, the following condition is imposed.

	\begin{assumption}~\label{assumption2}
		During the initial ESO transient, the actual disturbance and the disturbance-estimation error do not destabilize the closed-loop system or cause violations of the state constraints. Hence, the system remains bounded and safe, allowing the ESO to converge normally. After the transient, there exists a constant $\bar e>0$ such that the disturbance-estimation error satisfies
		$
		\|\m e_d(t)\|_2\le\bar e.
		$
	\end{assumption}

	Define the contracted safe set as
	\begin{equation*}
		\mathcal{C}_{\varepsilon}
		=
		\left\{
		\m{x}\in\mathbb{R}^{2n}
		\mid
		\m{q}_{\min}+\m{\varepsilon}
		\le\m{q}\le
		\m{q}_{\max}-\m{\varepsilon}
		\right\}.
	\end{equation*}

	
	\begin{theorem}\label{thm:forward}
		Suppose that Assumptions~\ref{assumption1}--\ref{assumption2} hold. 
		Then, there exists a positive constant $c_h > 0$, depending on the system dynamics and the barrier function, such that if the safety margin satisfies
		\begin{equation}
			(\m{\lambda}_1\odot\m{\lambda}_2)\odot\m{\varepsilon}
			\ge
			c_h \bar e\,\m{1}_n,
			\label{eq:margin_condition}
		\end{equation}
		and the initial state satisfies
		$
		\m{x}(0)
		\in
		\mathcal{C}_{\varepsilon},
		$
		then the closed-loop system under the QP controller \eqref{eq:qp_controller} satisfies
		$
		\m{x}(t)
		\in
		\mathcal{C},
		\forall t\ge0.
		$
	\end{theorem}
	
	\begin{remark}
		The detailed proof of Theorem~\ref{thm:forward} is provided in Appendix~A.
	\end{remark}

	For clarity, the complete implementation procedure of the proposed dynamic-constraint-reconstruction-based safety controller is summarized in Algorithm~\ref{alg:dcr_safe_control}.
	
	\begin{algorithm}[t]
		\caption{Dynamic Constraint Reconstruction Based Safety Control}
		\label{alg:dcr_safe_control}
		
		\KwData{
			Initial state $\m{x}(0)\in\mathcal C_\varepsilon$;
			observer state $\hat{\m{x}}_e(0)$;
			sampling period $\Delta t$;
			safety-margin vector $\m{\varepsilon}$ satisfying \eqref{eq:margin_condition}.
		}
		
		\KwResult{
			Safe control input $\m{u}^{*}(k)$.
		}
		
		Initialize ESO state $\hat{\m{x}}_e(0)$\;
		
		\For{$k=0,1,2,\ldots$}{
			
			Measure system state
			$\m{x}(k)=
			[\m{q}^{\top}(k),\dot{\m{q}}^{\top}(k)]^{\top}$\;
			
			Compute nominal tracking input
			$\m{u}_{\mathrm{des}}(k)$\;
			
			Update ESO according to \eqref{eq:eso_discrete}\;
			
			Extract disturbance estimate
			$\hat{\m{d}}(k)=\hat{\m{z}}(k)$\;
			
			Reconstruct system dynamics using
			\eqref{eq:reconstructed_acc} ~\cite{wang2023eso}\;
			
			Construct adaptive upper- and lower-bound HOCBF constraints using
			\eqref{eq:reconstructed_constraint} and
			\eqref{eq:reconstructed_lower_constraint}\;
			
			Introduce safety margin using
			\eqref{eq:margin_barrier}\;
			
			Update upper- and lower-bound robust safety constraints using
			\eqref{eq:robust_constraint}\;
			
			Construct admissible safe set
			$\mathcal U_{\mathrm{safe}}(k)$ from both constraints\;
			
			Solve the QP in \eqref{eq:qp_controller}
			to obtain $\m{u}^{*}(k)$\;
			
			Apply $\m{u}^{*}(k)$ to
			\eqref{eq:real_dynamics}\;
			
		}
		
	\end{algorithm}

	\section{Experimental Validation}
	\label{sec:ExperimentalValidation}
	
	This section experimentally validates the proposed dynamic-constraint-reconstruction-based safety controller developed in Algorithm~\ref{alg:dcr_safe_control}. According to the theoretical analysis in Sections~\ref{sec:dynamicCRF} and~\ref{sec:safetyCriticalCD}, the experiments are designed to verify the disturbance-estimation capability of the ESO in \eqref{eq:eso_discrete}, the forward-invariance guarantee established in Theorem~\ref{thm:forward}, and the conservatism-reduction capability of the reconstructed HOCBF constraints in \eqref{eq:robust_constraint}. Comparative studies with the standard CBF and robust CBF methods are conducted throughout all experiments.

	\subsection{Experimental Setup}
	\subsubsection{Excavation}
	
	Numerical simulations are conducted on a 4-DOF excavation robotic manipulator using the Robotics Toolbox~\cite{rtb}. The manipulator dynamics follow the disturbed model in \eqref{eq:real_dynamics}, where
	$\m q\in\mathbb R^4$
	denotes the joint-position vector.
	
	The unknown disturbance is selected as
	$
	\m d(t)
	=
	\m d_{\max}\cos(\omega t),
	$
	where
	$
	\m d_{\max}
	=
	[2.0,1.5,0.2,0.1]^\top\times10^6
	$
	and
	$\omega=0.02$. The desired trajectory is selected as
	$
	\m q_d(t)
	=
	\m A\sin(\omega_d t)+\m q_0,
	$
	where
	$
	\m A
	=
	[\pi,1.6,1.5,2.2]^\top
	$
	and
	$\omega_d=0.01$. The safety-margin vector in \eqref{eq:margin_barrier} is selected as
	$\m{\varepsilon}=[0.1,0.1,0.1,0.2]^\top$,
	which satisfies the condition in \eqref{eq:margin_condition}.
	For $n=4$, the augmented state comprises the joint position, joint velocity, and lumped disturbance, so the ESO is of order $3n=12$. The observer gain in \eqref{eq:eso_discrete} is calculated once from the discrete system matrices evaluated at the initial state and is subsequently used throughout the simulation. The resulting gain is denoted by $\m{L}_{k}$, whose scaled numerical value, reported to four decimal places, is
	\begin{equation*}
		\m{L}_{k}
		\simeq
		10^{7}
		\left[
		\begin{smallmatrix}
		 0.0000 &  0.0000 &  0.0000 &  0.0000\\
		 0.0000 &  0.0000 &  0.0000 &  0.0000\\
		 0.0000 &  0.0000 &  0.0000 &  0.0000\\
		 0.0000 &  0.0000 &  0.0000 &  0.0000\\
		 0.0000 &  0.0000 &  0.0000 &  0.0000\\
		 0.0000 &  0.0000 &  0.0000 &  0.0000\\
		 0.0000 &  0.0000 &  0.0000 &  0.0000\\
		 0.0000 &  0.0000 &  0.0000 &  0.0000\\
		-2.3477 &  0.0000 &  0.0000 &  0.0000\\
		 0.0000 & -2.2547 & -0.9988 & -1.3070\\
		 0.0000 & -0.2456 & -0.7199 & -0.4020\\
		 0.0000 &  0.0220 & -0.1847 & -0.1217
		\end{smallmatrix}
		\right]
		\in\mathbb{R}^{12\times 4}.
	\end{equation*}
	The ESO is implemented according to \eqref{eq:eso_discrete}, and the safe control input is generated online using Algorithm~\ref{alg:dcr_safe_control}. According to~\eqref{eq:nominal_hocbf}, the gain vectors are $\m{\lambda}_1=[20,20,20,20]^\top$ and $\m{\lambda}_2=[40,40,40,40]^\top$. The total simulation duration is $T=16$ s with sampling interval $\Delta t=0.01$ s. Three control methods are compared, namely the standard CBF, the robust CBF, and the proposed DCR-CBF.

	For a fair comparison, the standard CBF and robust CBF use the same nominal tracking controller, HOCBF gain vectors, safety margin, joint limits, actuator limits, QP settings, initial conditions, and sampling interval as the proposed DCR-CBF. The three methods differ only in their treatment of the unknown disturbance. The standard CBF constructs its constraint from the nominal zero-disturbance model. The robust CBF assumes the componentwise worst-case bound $|\m d(t)|\le\m d_{\max}$ and uses the maximum disturbance magnitude $\m d_{\max}$ to tighten its safety constraint at every sampling instant. In contrast, DCR-CBF reconstructs the constraint online using the current ESO estimate $\hat{\m d}(t)$.
	
		\subsubsection{KUKA iiwa7}
	
	To further evaluate the scalability of the proposed framework to 
	higher dimensional systems, numerical simulations are also conducted 
	on a 7-DOF KUKA iiwa7 manipulator. The manipulator dynamics follow the 
	disturbed model in \eqref{eq:real_dynamics}, where $\m q\in\mathbb R^7$ 
	denotes the joint-position vector.
	
	It is worth emphasizing that each joint of the real iiwa7 is equipped 
	with a high-ratio harmonic drive, whose rotor inertia, when reflected to 
	the joint side via the square of the gear ratio, is far from negligible. 
	To capture this characteristic, the nominal inertia matrix is constructed 
	by augmenting the link-side inertia with a joint-side reflected-inertia 
	matrix $\m J_{\mathrm{refl}}\in\mathbb R^{7\times 7}$, i.e.,
	$
	\m M(\m q)
	=
	\m M_{\mathrm{link}}(\m q)
	+
	\m J_{\mathrm{refl}},
	$
	where
	$
	\m J_{\mathrm{refl}}
	=
	\mathrm{diag}
	\bigl([0.85,0.85,0.45,0.45,0.15,0.15,0.05]^\top\bigr).
	$
	Neglecting this term renders the equivalent inertia of the distal joints 
	smaller than the true value by two to three orders of magnitude, leading 
	to an ill-conditioned dynamics and severely degraded control performance. 
	The residual discrepancies between the nominal model and the actual system, 
	including parameter variations, unmodeled dynamics, and external 
	disturbances, are lumped into the disturbance term $\m d(t)$.
	
	The unknown disturbance is selected as
	$
	\m d(t)
	=
	\m d_{\max}\sin(\omega t),
	$
	with
	$
	\m d_{\max}
	=
	[200,200,200,10,20,20,10]^\top
	$
	and $\omega=1.0$~rad/s. Since the actuator ratings and equivalent inertia 
	of the last four joints are markedly lower than those of the first three, 
	applying the same disturbance magnitude to all joints would render the 
	worst-case constraints in the robust CBF infeasible. The disturbance 
	amplitudes for the last four joints are therefore scaled down accordingly, 
	so that the three methods can all operate within their respective feasible 
	domains, thereby isolating the structural difference among the constraint 
	strategies. The desired trajectory is selected as
	$
	\m q_d(t)
	=
	\m A\sin(\omega_d t),
	$
	where
	$
	\m A
	=
	1.2\cdot\m 1_7
	$
	and $\omega_d=0.4$~rad/s. The safety-margin vector in 
	\eqref{eq:margin_barrier} is chosen as
	$
	\m{\varepsilon}
	=
	[0.1,0.05,0.15,0.015,0.1,0.1,0.1]^\top.
	$
	
	For $n=7$, the augmented state comprises joint positions, joint 
	velocities, and the lumped disturbance, so the ESO is of order $3n=21$. 
	The observer gain in \eqref{eq:eso_discrete} is computed once from the 
	discrete system matrices evaluated at the initial state and is subsequently 
	held constant throughout the simulation. The ESO is implemented according 
	to \eqref{eq:eso_discrete}, and the safe control input is generated online 
	using Algorithm~\ref{alg:dcr_safe_control}. According to 
	\eqref{eq:nominal_hocbf}, the gain vectors are
	$\m{\lambda}_1=[20,\ldots,20]^\top\in\mathbb R^7$ and
	$\m{\lambda}_2=[40,\ldots,40]^\top\in\mathbb R^7$. The total simulation 
	duration is $T=40$~s with sampling interval $\Delta t=0.01$~s.
	
	The compared methods, fairness settings, and the treatment of the unknown 
	disturbance by each strategy are identical to those described for the 
	4-DOF excavator in Section~\ref{sec:ExperimentalValidation} and are 
	therefore not repeated here.
	
	\subsection{Disturbance Estimation Performance}
	
	This experiment validates the disturbance-estimation property assumed in Assumption~\ref{assumption2}. Fig.~\ref{fig:u_tracking} presents the disturbance-estimation results obtained by the ESO in \eqref{eq:eso_discrete}. It can be observed that the estimated disturbance $\hat{\m d}(t)$ rapidly converges to the true disturbance after a short transient phase and accurately tracks the time-varying disturbance in all joint channels. These results experimentally verify that the disturbance-estimation error remains bounded, thereby supporting Assumption~\ref{assumption2} and providing reliable information for the dynamic constraint reconstruction in \eqref{eq:reconstructed_constraint}.
	
	\begin{figure*}[t]
		\centering
		
		\begin{minipage}{0.48\textwidth}
			\centering
			\includegraphics[width=0.95\linewidth]{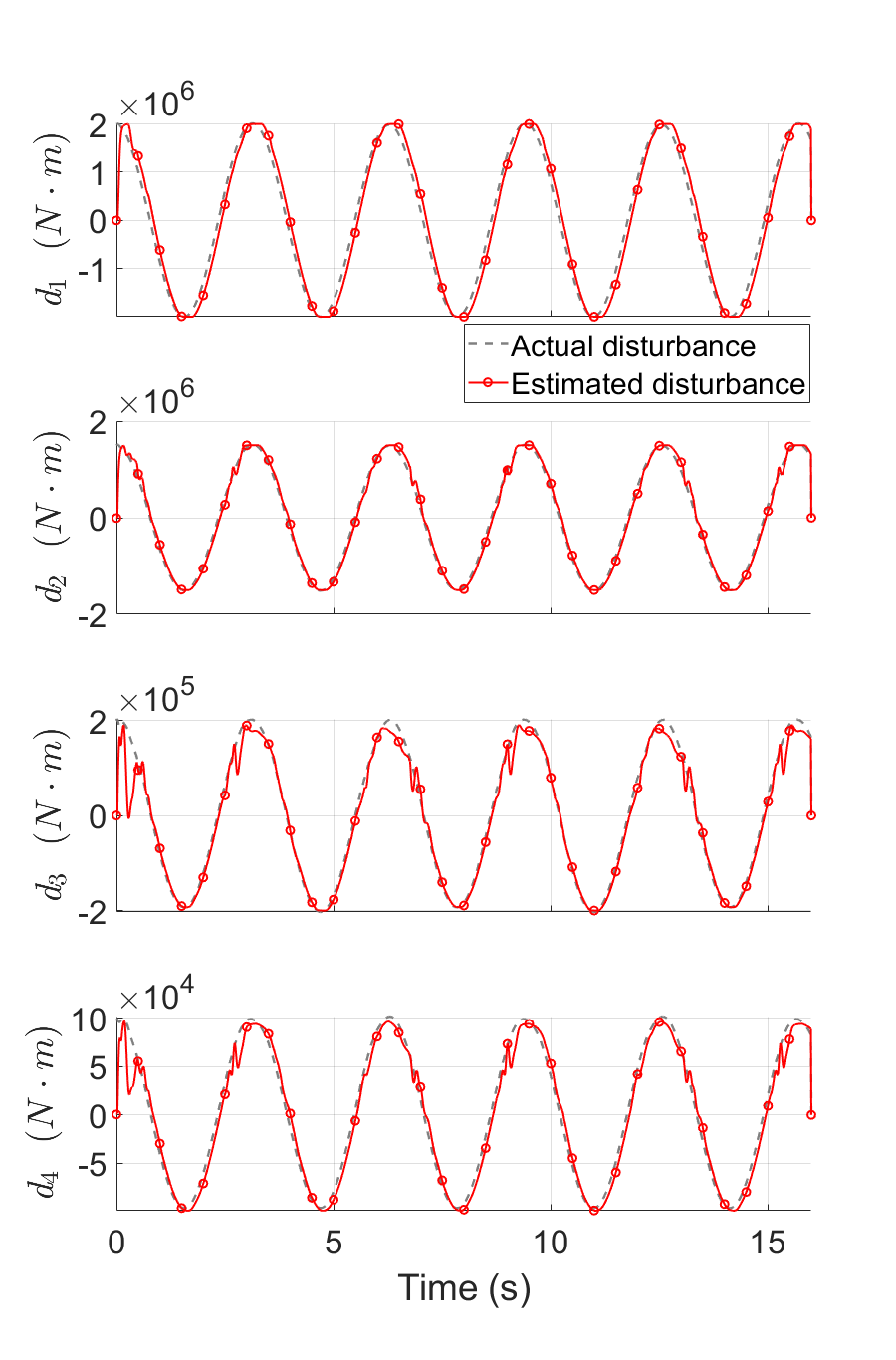}
			\caption{Online disturbance estimation of the ESO.}
			\label{fig:u_tracking}
		\end{minipage}
		\hfill  
		\begin{minipage}{0.48\textwidth}
			\centering
			\includegraphics[width=0.95\linewidth]{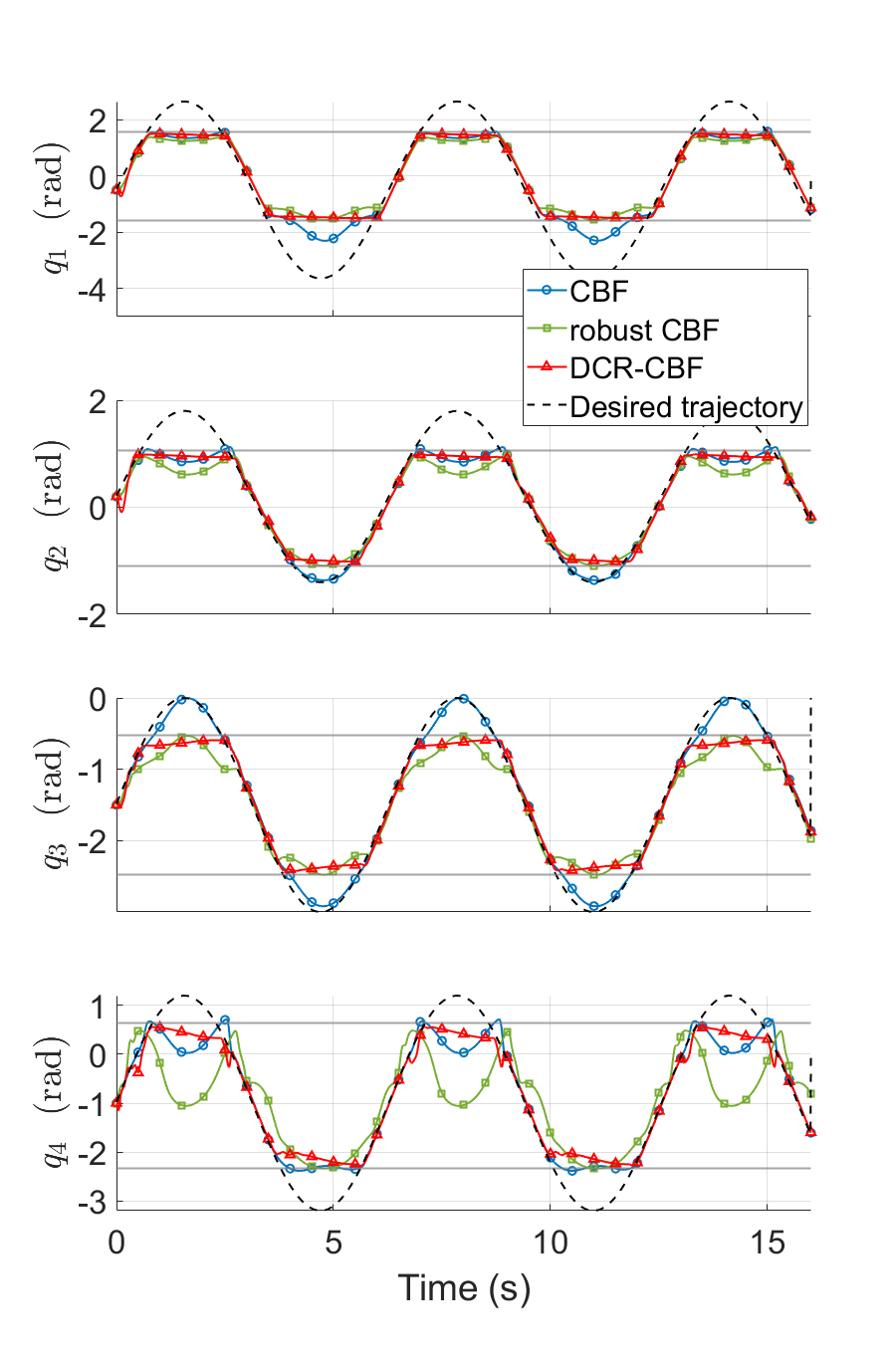}
			\caption{Joint trajectories under different methods.}
			\label{fig:tracking}
		\end{minipage}
		
	\end{figure*}

		\subsection{Safety Constraint Satisfaction}
	
	This experiment validates the forward-invariance result established in Theorem~\ref{thm:forward}.
	
	Fig.~\ref{fig:tracking} shows the joint trajectories together with the 
	safety boundaries for the 4-DOF excavator, and Fig.~\ref{fig:iiwa7_tracking} 
	presents the corresponding results for the 7-DOF iiwa7 manipulator. On 
	both platforms, the standard CBF exhibits multiple boundary violations 
	under unknown disturbances, indicating that the fixed nominal constraints 
	derived from \eqref{eq:nominal_constraint} become inconsistent with the 
	actual disturbed dynamics. Although the robust CBF successfully preserves 
	safety, its trajectories remain significantly away from the safety 
	boundaries due to worst-case constraint tightening. In contrast, the 
	proposed DCR-CBF maintains all trajectories strictly inside the safe 
	region under the reconstructed safe set $\mathcal U_{\mathrm{safe}}(t)$ 
	defined in Section~\ref{sec:redundantSMD}, which experimentally validates 
	the forward-invariance guarantee of Theorem~\ref{thm:forward} on both 
	the 4-DOF and 7-DOF systems.
	
	\begin{figure*}[t]
		\centering
		\includegraphics[width=0.95\textwidth]{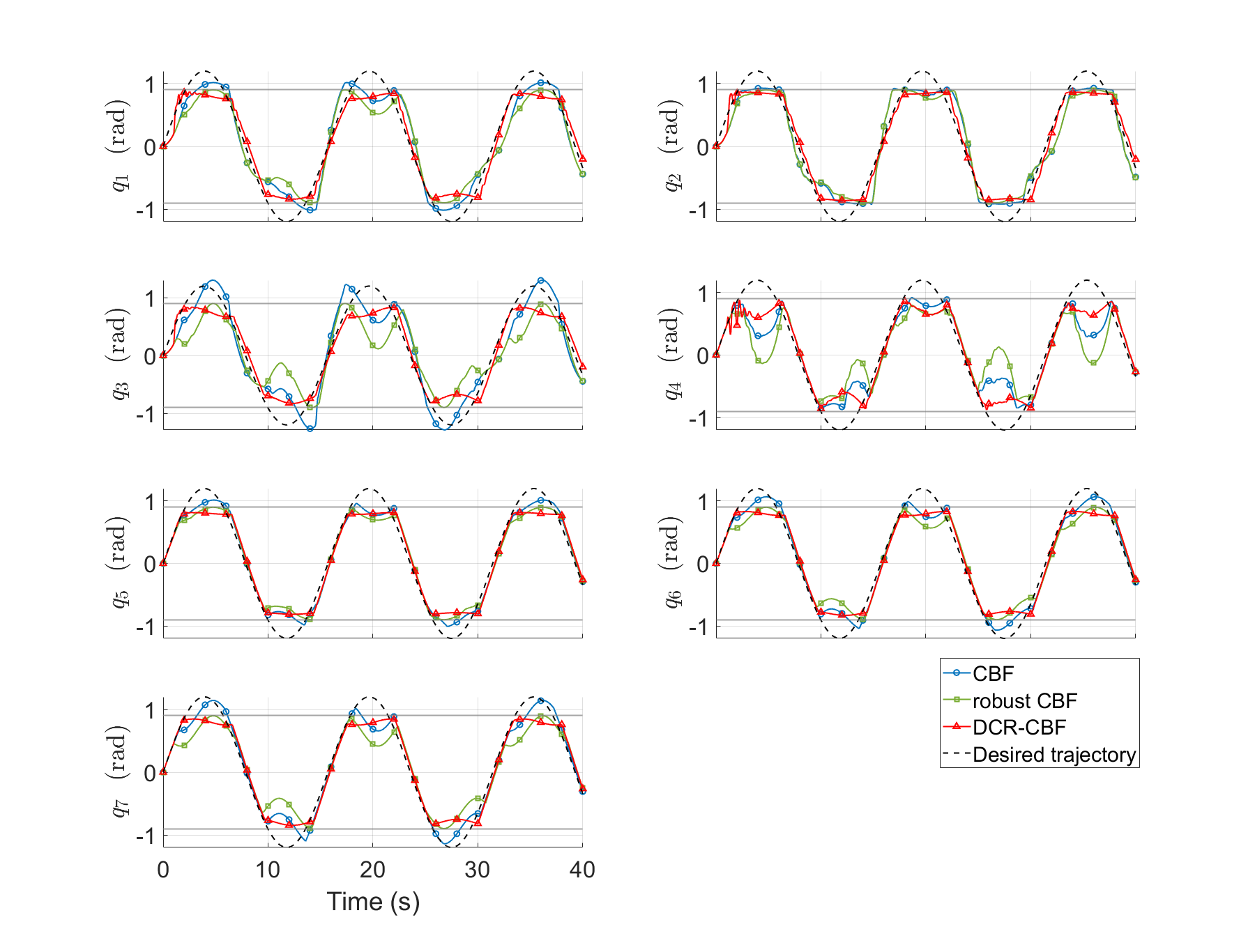}
		\caption{Joint trajectories of the 7-DOF iiwa7 manipulator under 
			the three control methods. The dashed black curve is the desired 
			trajectory, and the shaded region bounded by the solid gray lines 
			denotes the safe set. The standard CBF violates the safety boundary, 
			the robust CBF remains overly conservative, and the proposed DCR-CBF 
			achieves zero violation with the best tracking accuracy.}
		\label{fig:iiwa7_tracking}
	\end{figure*}

		\subsection{Tracking Performance and Conservatism Analysis}
	
	In addition to safety performance, trajectory-tracking performance is 
	also evaluated. As shown in Fig.~\ref{fig:tracking} and 
	Fig.~\ref{fig:iiwa7_tracking}, the robust CBF significantly sacrifices 
	tracking accuracy due to conservative worst-case constraint design. 
	By contrast, the proposed method continuously updates the HOCBF 
	constraints according to the disturbance estimate using 
	\eqref{eq:robust_constraint}, which allows the QP controller in 
	\eqref{eq:qp_controller} to generate less conservative control inputs. 
	Consequently, the proposed method achieves significantly improved 
	trajectory tracking while preserving strict safety guarantees on both 
	the 4-DOF and 7-DOF systems.

		\subsection{Quantitative Comparison}
	
	To quantitatively evaluate both safety and tracking performance, the 
	root-mean-square tracking error and safety-violation rate are defined as
	$
	\mathrm{RMSE}
	=
	\sqrt{
		\frac1T
		\int_0^T
		\|
		\m q(t)-\m q_d(t)
		\|^2dt
	},
	$
	and
	$
	\mathrm{Violation}
	=
	\frac1T
	\int_0^T
	\mathbb I
	(
	\m q(t)\notin
	[\m q_{\min},\m q_{\max}]
	)dt.
	$
	
	\begin{table}[t]
		\centering
		\caption{Quantitative Performance Comparison on the 4-DOF Excavator.}
		\label{tab:performance}
		
		\setlength{\tabcolsep}{4pt}
		
		\begin{tabular}{c|c|c|c|c|c}
			\hline
			\textit{Method}
			& \textit{Joint}
			& \textit{RMSE}
			& \textit{Viol.}
			& \textit{Avg RMSE}
			& \textit{Avg Viol.}
			\\
			\hline
			
			\multirow{4}{*}{CBF}
			& J1 & 0.2811 & 0.2055 & \multirow{4}{*}{0.2279} & \multirow{4}{*}{0.2748} \\
			& J2 & 0.1286 & 0.2461 & & \\
			& J3 & 0.2499 & 0.4909 & & \\
			& J4 & 0.2521 & 0.1568 & & \\
			\hline
			
			\multirow{4}{*}{Robust CBF}
			& J1 & 0.2702 & 0.0000 & \multirow{4}{*}{0.3741} & \multirow{4}{*}{0.0000} \\
			& J2 & 0.2123 & 0.0000 & & \\
			& J3 & 0.1887 & 0.0000 & & \\
			& J4 & 0.8251 & 0.0000 & & \\
			\hline
			
			\multirow{4}{*}{DCR-CBF}
			& J1 & \textbf{0.1664} & 0.0000 &
			\multirow{4}{*}{\textbf{0.1398}} &
			\multirow{4}{*}{0.0000} \\
			& J2 & \textbf{0.1075} & 0.0000 & & \\
			& J3 & \textbf{0.0989} & 0.0000 & & \\
			& J4 & \textbf{0.1863} & 0.0000 & & \\
			\hline
		\end{tabular}
	\end{table}
	
	Table~\ref{tab:performance} summarizes both joint-wise and overall 
	performance metrics for the 4-DOF excavator. The relatively large 
	robust-CBF RMSE for J4 (0.8251) results from the comparatively large 
	effective maximum disturbance bound in the J4 channel. Since the robust 
	CBF applies this worst-case bound at every sampling instant, the 
	admissible safe region for J4 is shifted downward and substantially 
	contracted, forcing the joint trajectory away from its reference and 
	thereby increasing the RMSE despite zero safety violations. The proposed 
	DCR-CBF achieves zero safety violation on all joints while obtaining the 
	lowest average RMSE among all methods. Compared with the robust CBF, the 
	proposed method significantly reduces conservatism through online dynamic 
	constraint reconstruction. Compared with the standard CBF, it successfully 
	preserves forward invariance under unknown disturbances.
	
	\begin{table}[t]
		\centering
		\caption{Quantitative Performance Comparison on the 7-DOF iiwa7 Manipulator.}
		\label{tab:performance_iiwa7}
		
		\setlength{\tabcolsep}{4pt}
		
		\begin{tabular}{c|c|c|c|c|c}
			\hline
			\textit{Method}
			& \textit{Joint}
			& \textit{RMSE}
			& \textit{Viol.}
			& \textit{Avg RMSE}
			& \textit{Avg Viol.}
			\\
			\hline
			
			\multirow{7}{*}{CBF}
			& J1 & 0.1876 & 0.3292 & \multirow{7}{*}{0.1754} & \multirow{7}{*}{0.2588} \\
			& J2 & 0.1771 & 0.3502 & & \\
			& J3 & 0.2629 & 0.3389 & & \\
			& J4 & 0.2667 & 0.0100 & & \\
			& J5 & \textbf{0.0925} & 0.2387 & & \\
			& J6 & 0.1087 & 0.2729 & & \\
			& J7 & 0.1327 & 0.2719 & & \\
			\hline
			
			\multirow{7}{*}{Robust CBF}
			& J1 & 0.2217 & 0.0000 & \multirow{7}{*}{0.2471} & \multirow{7}{*}{0.0000} \\
			& J2 & 0.1863 & 0.0000 & & \\
			& J3 & 0.3660 & 0.0000 & & \\
			& J4 & 0.4572 & 0.0000 & & \\
			& J5 & 0.1157 & 0.0000 & & \\
			& J6 & 0.1598 & 0.0000 & & \\
			& J7 & 0.2227 & 0.0000 & & \\
			\hline
			
			\multirow{7}{*}{DCR-CBF}
			& J1 & \textbf{0.1111} & 0.0000 & \multirow{7}{*}{\textbf{0.1149}} & \multirow{7}{*}{0.0000} \\
			& J2 & \textbf{0.0903} & 0.0000 & & \\
			& J3 & \textbf{0.1392} & 0.0000 & & \\
			& J4 & \textbf{0.1598} & 0.0000 & & \\
			& J5 & 0.1010 & 0.0000 & & \\
			& J6 & \textbf{0.1019} & 0.0000 & & \\
			& J7 & \textbf{0.1009} & 0.0000 & & \\
			\hline
		\end{tabular}
	\end{table}
	
	Table~\ref{tab:performance_iiwa7} summarizes the joint-wise and overall 
	performance metrics for the 7-DOF iiwa7 manipulator. The standard CBF 
	exhibits significant safety violations on six out of seven joints, with 
	an average violation rate of 25.88\%, confirming that the fixed nominal 
	constraints become increasingly inconsistent with the actual disturbed 
	dynamics as the system dimension grows. The exceptionally low violation 
	rate on J4 of the standard CBF (0.01) is due to the comparatively small 
	disturbance amplitude assigned to that channel (10~N$\cdot$m versus 
	200~N$\cdot$m for J1--J3), rather than to the effectiveness of the 
	nominal constraint. The robust CBF achieves zero violation by applying 
	the componentwise worst-case bound at every sampling instant, but its 
	average RMSE (0.2471) is the highest among the three methods; in 
	particular, the RMSE of J4 reaches 0.4572, since the comparatively small 
	equivalent inertia in that channel forces the admissible safe region to 
	contract substantially. In contrast, the proposed DCR-CBF achieves zero 
	safety violation on all seven joints while attaining the lowest average 
	RMSE (0.1149), reducing the average RMSE by approximately 53\% compared 
	with the robust CBF and by 35\% compared with the standard CBF. 
	
	Across both platforms, the proposed DCR-CBF consistently achieves zero 
	safety violation and the lowest average RMSE, demonstrating that the 
	dynamic constraint reconstruction framework scales effectively from a 
	4-DOF excavator to a 7-DOF manipulator while reducing conservatism 
	without compromising safety.

		\section{Conclusion}
	\label{sec:Conclusion}
	
	This paper proposed a dynamic-constraint-reconstruction-based safety 
	control framework for high-dimensional robotic manipulators subject to 
	unknown disturbances and model uncertainties. Unlike conventional CBF 
	methods that rely on fixed nominal constraints, the proposed method 
	reconstructs high-order safety constraints online according to the 
	estimated true system dynamics using an extended state observer. To 
	address disturbance-estimation inaccuracies, a safety-margin design was 
	further introduced, and a sufficient condition was established to 
	guarantee forward invariance under bounded estimation errors.
	
	Numerical simulations on a 4-DOF excavation manipulator and a 7-DOF 
	KUKA iiwa7 manipulator validated the effectiveness of the proposed 
	framework. The proposed DCR-CBF method achieves zero safety violation 
	under strong unknown disturbances while significantly improving 
	trajectory-tracking performance compared with the standard CBF and the 
	robust CBF. On the 7-DOF iiwa7 manipulator, whose nominal model 
	explicitly incorporates the joint-side reflected inertia of the 
	harmonic drives, the proposed method attains an average RMSE of 0.1149 
	with zero violation, reducing the average RMSE by approximately 53\% 
	and 35\% compared with the robust CBF and the standard CBF, 
	respectively. The standard CBF fails to preserve forward invariance 
	under unknown disturbances, whereas the robust CBF guarantees safety 
	at the cost of excessive conservatism.
	
	Future work will focus on hardware implementation on hydraulic 
	excavation platforms and on the integration of learning-based 
	disturbance prediction for more complex contact-rich environments.

	\appendices
	\section{Proof of Theorem~\ref{thm:forward}}
	\label{appendix:proof}
	
	\begin{proof}
		Consider the upper-bound barrier $h(\m{x})=\m{q}_{\max}-\m{q}$ in \eqref{eq:joint_barrier}. 
		
		The true disturbance is decomposed as
		$
		\m{d}(t)=\hat{\m{d}}(t)+\m{e}_d(t),
		$
		where $\m{e}_d(t)$ is the disturbance-estimation error. Define the acceleration reconstructed using the ESO estimate as
		\begin{equation*}
			\ddot{\m{q}}_{\mathrm{rec}}
			=
			\m{M}^{-1}(\m{q})
			\Bigl(
			\m{u}+\hat{\m{d}}(t)
			-\m{C}(\m{q},\dot{\m{q}})\dot{\m{q}}
			-\m{G}(\m{q})
			\Bigr).
		\end{equation*}
		It follows from \eqref{eq:real_dynamics} that
		$
			\ddot{\m{q}}
			=
			\ddot{\m{q}}_{\mathrm{rec}}
			+
			\m{M}^{-1}(\m{q})\m{e}_d(t).
		$
		
		Because $\dot h=-\dot{\m{q}}$, the actual and reconstructed second derivatives of the upper-bound barrier satisfy
		\begin{equation*}
			\ddot h
			=
			\ddot h_{\mathrm{rec}}
			-
			\m{M}^{-1}(\m{q})\m{e}_d(t),
			\qquad
			\ddot h_{\mathrm{rec}}=-\ddot{\m{q}}_{\mathrm{rec}}.
		\end{equation*}
		Therefore, the disturbance-estimation error affects the reconstructed HOCBF directly as
		\begin{equation*}
			\begin{aligned}
				&\ddot h
				+(\m{\lambda}_1+\m{\lambda}_2)\odot\dot h
				+(\m{\lambda}_1\odot\m{\lambda}_2)\odot h\\
				={}&
				\ddot h_{\mathrm{rec}}
				+(\m{\lambda}_1+\m{\lambda}_2)\odot\dot h
				+(\m{\lambda}_1\odot\m{\lambda}_2)\odot h
				-\m{M}^{-1}(\m{q})\m{e}_d(t).
			\end{aligned}
		\end{equation*}
		Thus, the only discrepancy between the reconstructed and actual HOCBF conditions is the additive term $-\m{M}^{-1}(\m{q})\m{e}_d(t)$.

		After the ESO has entered the bounded-error regime in Assumption~\ref{assumption2}, $\|\m{e}_d(t)\|_2\le\bar e$. Define
		\begin{equation*}
			c_h
			:=
			\max_{i=1,\ldots,n}
			\sup_{\m{x}\in\mathcal X}
			\left\|
			\left[\m{M}^{-1}(\m{q})\right]_{i,:}
			\right\|_2.
		\end{equation*}
		The constant $c_h$ is finite under Assumption~\ref{assumption1}, and the $i$th component of the HOCBF perturbation satisfies
		\begin{equation*}
			\left|
			\left[\m{M}^{-1}(\m{q})\m{e}_d(t)\right]_i
			\right|
			\le
			c_h\bar e.
		\end{equation*}

		The upper-bound inequality in \eqref{eq:robust_constraint}, obtained from the boundary contraction in \eqref{eq:margin_barrier}, is equivalent to enforcing
		\begin{equation*}
			\ddot h_{\mathrm{rec}}
			+(\m{\lambda}_1+\m{\lambda}_2)\odot\dot h
			+(\m{\lambda}_1\odot\m{\lambda}_2)\odot h
			\ge
			(\m{\lambda}_1\odot\m{\lambda}_2)\odot\m{\varepsilon}.
		\end{equation*}
		Combining this safety margin with the estimation-error bound gives, for each joint,
		\begin{equation*}
			\begin{aligned}
				\ddot h_i
				&+(\lambda_{1,i}+\lambda_{2,i})\dot h_i
				+\lambda_{1,i}\lambda_{2,i}h_i
				\ge
				\lambda_{1,i}\lambda_{2,i}\varepsilon_i-c_h\bar e.
			\end{aligned}
		\end{equation*}
		If
		\begin{equation*}
			(\m{\lambda}_1\odot\m{\lambda}_2)\odot\m{\varepsilon}
			\ge
			c_h\bar e\,\m{1}_n,
		\end{equation*}
		then the actual upper-bound HOCBF condition remains nonnegative despite the disturbance-estimation error. The lower-bound constraint follows from the same argument: its estimation-error term has the opposite sign but the same bound $c_h\bar e$, so the same safety-margin condition applies. Consequently, under the standard HOCBF feasibility and initial-condition requirements, both joint bounds are forward invariant after the ESO enters its bounded-error regime. Assumption~\ref{assumption2} supplies constraint satisfaction during the initial ESO transient, and hence
		$
		\m{x}(t)\in\mathcal C, \forall t\ge0.
		$
		This completes the proof.
	\end{proof}

	\section{Numerical Verification of the Margin Condition}
	\label{appendix:numerical}
	
	To verify the theoretical derivation of the safety margin condition
	\eqref{eq:margin_condition} in Theorem~\ref{thm:forward}, additional
	numerical experiments are designed. In practical systems, the exact safety
	margin is difficult to obtain online due to model uncertainties and unknown
	external disturbances. In simulation, however, the true states and
	disturbances are known, so the most accurate safety margin can be computed
	in real time, thereby allowing the condition to be examined. Specifically,
	at each sampling instant, $\m{\varepsilon}(k)$ is computed as
	\begin{equation}
		\m{\varepsilon}(k)
		=
		\left|
		(\m{\lambda}_1\odot\m{\lambda}_2)^{-1}
		\m{M}^{-1}\!\bigl(\m{q}(k)\bigr)\,
		\m{e}_d(k)
		\right|,
		\label{eq:online_margin}
	\end{equation}
	where $\m{e}_d(k)=\m{d}(k)-\hat{\m{d}}(k)$ is the disturbance-estimation
	error at instant $k$, obtained as the difference between the true
	disturbance and the ESO estimate, and $\m{M}\!\bigl(\m{q}(k)\bigr)$ is the
	inertia matrix evaluated at the current configuration. Since both
	quantities are available online in the simulation, \eqref{eq:online_margin}
	realizes the componentwise margin condition \eqref{eq:margin_condition}
	with the time-varying bound $\|\m{e}_d(k)\|_2\le\bar e$ and the
	configuration-dependent constant $c_h$.
	
	The experimental parameters are set as follows. To exercise the margin
	condition under different disturbance rates, the disturbance amplitude is
	kept identical to that in the main experiment, i.e.,
	$\m{d}_{\max}=[2.0,1.5,0.2,0.1]^{\top}\times10^{6}$, while the temporal
	profile is replaced by a periodic staircase triangular wave:
	$$
	d(k)
	=
	\m{d}_{\max}
	\left(
	1-2
	\left|
	\frac{\ell \bmod 2P}{P}-1
	\right|
	\right),
	\qquad
	\ell
	=
	\left\lfloor
	\frac{k}{\kappa}
	\right\rfloor,
	$$
	where $\kappa=200$ is the update interval in sampling steps and the
	sampling period is $\Delta t=0.01$~s. The half-period length $P$ is chosen
	as $2$, $4$, and $6$, respectively. The corresponding switching period and
	switching frequency are
	$$
	T=2\kappa P\Delta t,
	\qquad
	f=\frac{1}{2\kappa P\Delta t},
	$$
	which gives $T=8$~s ($f=0.1250$~Hz), $T=16$~s ($f=0.0625$~Hz), and
	$T=24$~s ($f=0.0417$~Hz). The safety margin $\m{\varepsilon}(k)$ is updated
	online at each sampling instant according to \eqref{eq:online_margin}, so
	that it always reflects the current estimation error.
	
	Fig.~\ref{fig:proof_dis} shows the disturbance estimation of the ESO for
	$P=2$, and Fig.~\ref{fig:proof_x_tracking} shows the joint trajectories
	under the three switching frequencies. The quantitative metrics for each
	frequency are reported in Table~\ref{tab:proof_dis}. It can be seen that
	under all three disturbance switching frequencies, all joint trajectories
	remain strictly inside the safe region with zero violation, which verifies
	the effectiveness of the online margin condition \eqref{eq:margin_condition}.
	Meanwhile, the value of $P$ affects the trajectory-tracking performance.
	This demonstrates that obtaining a safety margin satisfying the condition
	can guarantee the safety of the manipulator.
	
	\begin{figure*}[t]
		\centering
		\begin{minipage}{0.48\textwidth}
			\centering
			\includegraphics[width=0.95\linewidth]{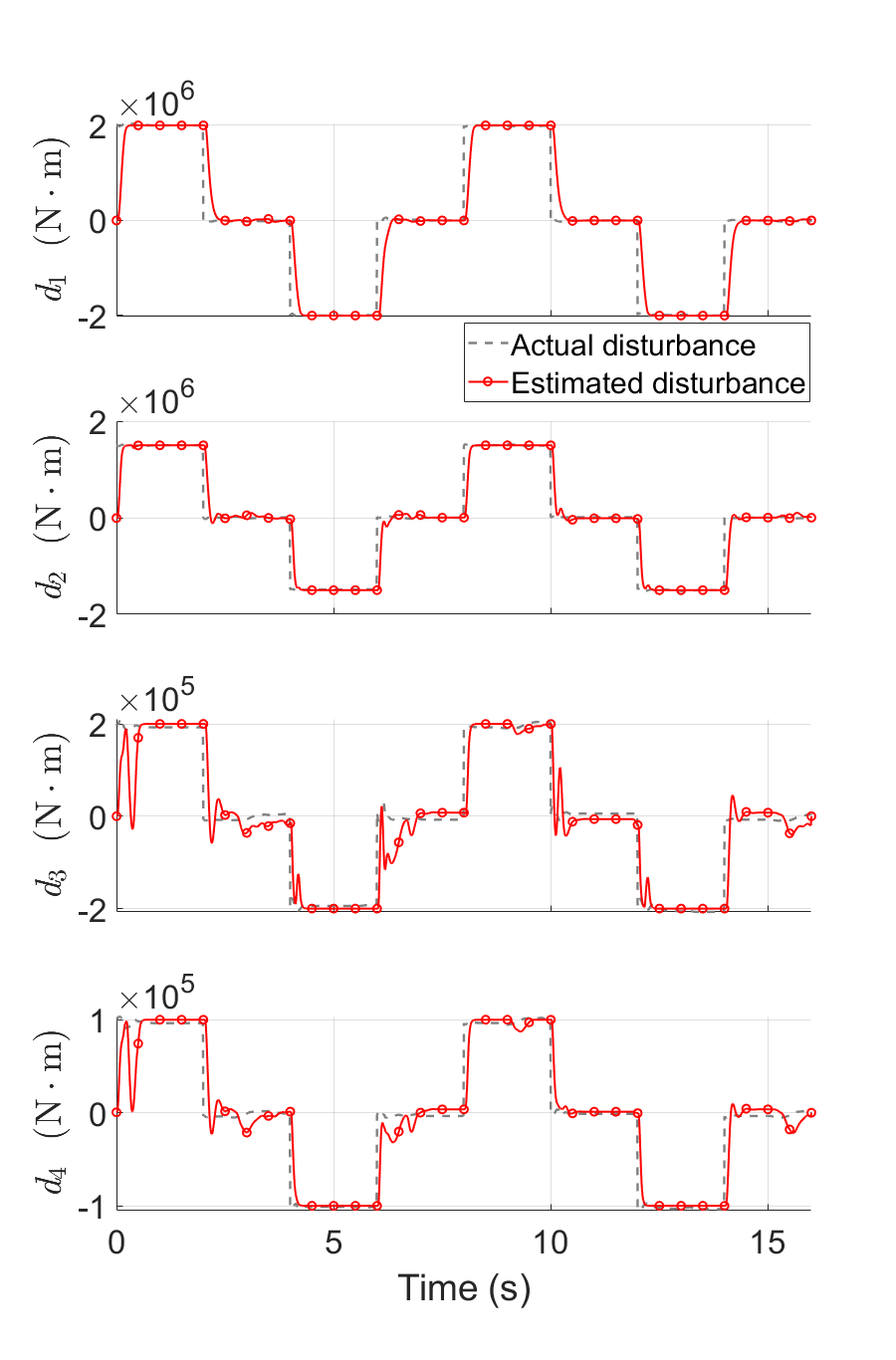}
			\caption{Online disturbance estimation of the ESO for $P=2$.}
			\label{fig:proof_dis}
		\end{minipage}
		\hfill
		\begin{minipage}{0.48\textwidth}
			\centering
			\includegraphics[width=0.95\linewidth]{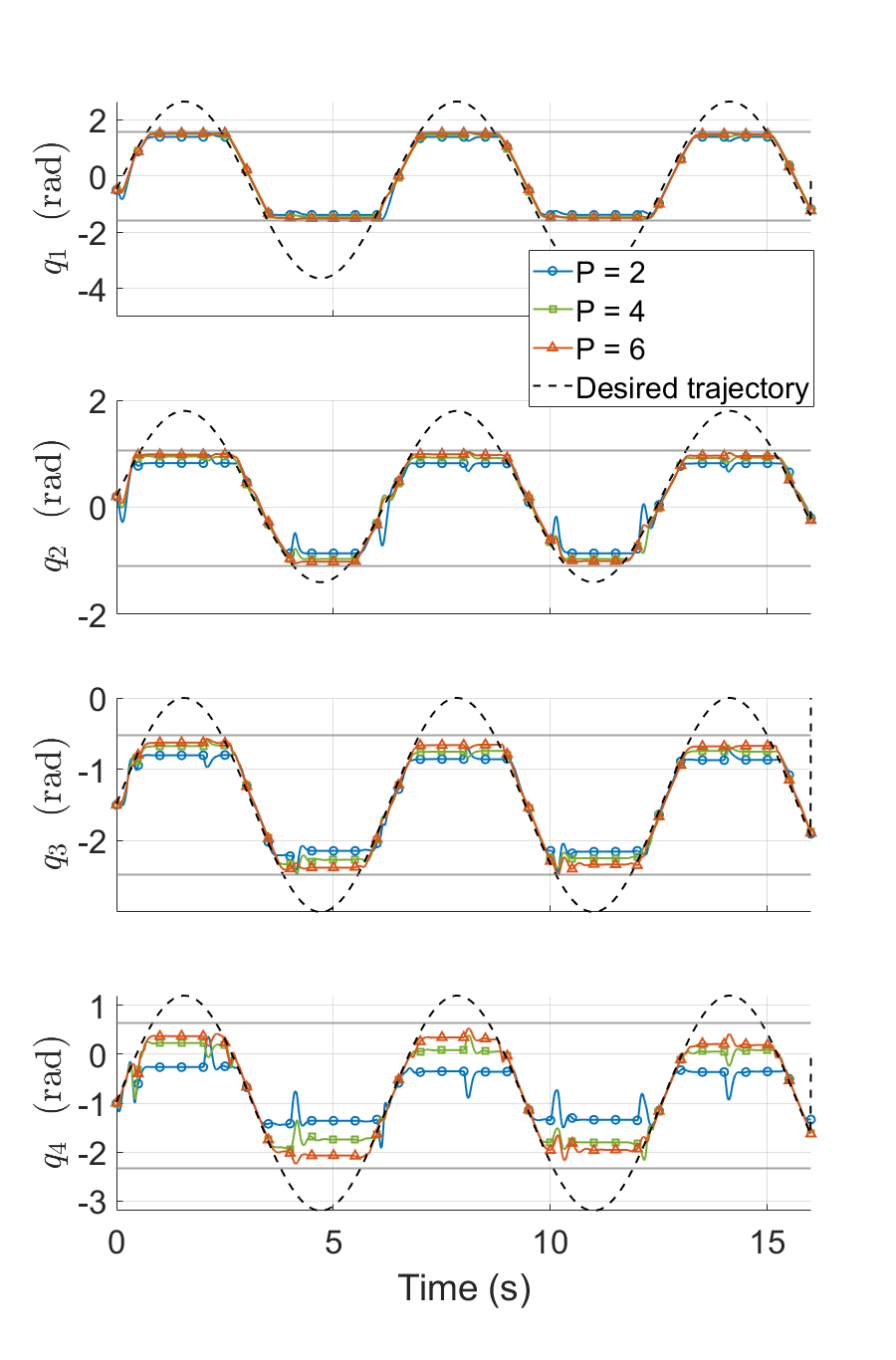}
			\caption{Joint trajectories under different disturbance switching frequencies.}
			\label{fig:proof_x_tracking}
		\end{minipage}
	\end{figure*}
	
	\begin{table}[t]
		\centering
		\caption{Quantitative tracking performance under different disturbance
			switching frequencies. All cases achieve zero safety violation.}
		\label{tab:proof_dis}
		\setlength{\tabcolsep}{4pt}
		\begin{tabular}{c|c|c|c|c|c|c}
			\hline
			\textit{Setting}
			&
			\textit{Joint}
			&
			\textit{RMSE}
			&
			\textit{Viol.}
			&
			\textit{Avg RMSE}
			&
			\textit{Avg Viol.}
			&
			$\varepsilon_i(T)$
			\\
			\hline
			
			\multirow{4}{*}{$P=2$}
			
			& J1 & 0.2387 & 0.0000 &
			\multirow{4}{*}{0.3800} &
			\multirow{4}{*}{0.0000} &
			0.1861 \\
			
			& J2 & 0.2259 & 0.0000 & & & 0.2346 \\
			
			& J3 & 0.2519 & 0.0000 & & & 0.3384 \\
			
			& J4 & 0.8032 & 0.0000 & & & 1.0007 \\
			
			\hline
			
			\multirow{4}{*}{$P=4$}
			
			& J1 & 0.1756 & 0.0000 &
			\multirow{4}{*}{0.2249} &
			\multirow{4}{*}{0.0000} &
			0.1042 \\
			
			& J2 & 0.1283 & 0.0000 & & & 0.1292 \\
			
			& J3 & 0.1682 & 0.0000 & & & 0.2234 \\
			
			& J4 & 0.4277 & 0.0000 & & & 0.5623 \\
			
			\hline
			
			\multirow{4}{*}{$P=6$}
			
			& J1 & \textbf{0.1664} & 0.0000 &
			\multirow{4}{*}{\textbf{0.1602}} &
			\multirow{4}{*}{0.0000} &
			\textbf{0.0772} \\
			
			& J2 & \textbf{0.0927} & 0.0000 & & & \textbf{0.0879} \\
			
			& J3 & \textbf{0.1137} & 0.0000 & & & \textbf{0.1554} \\
			
			& J4 & \textbf{0.2681} & 0.0000 & & & \textbf{0.4178} \\
			
			\hline
			
		\end{tabular}
	\end{table}

	\bibliography{reference}

\end{document}